\documentclass[reqno, centertags, 11pt]{amsart}
\usepackage{amsmath,amsthm,amsfonts,amssymb}
\usepackage{url}
\usepackage{graphicx}
\usepackage{epstopdf}
\usepackage[dvipsnames]{xcolor}
\usepackage{xcolor}
\usepackage{comment}

\DeclareMathOperator*{\argmax}{arg\,max}

\renewcommand{\part}[1]{\left(#1\right)}  
\newcommand{\Tr}{{ \rm tr}}       
\newcommand{\Prob}{{ \rm p}}   

\newcommand{\N}{\mathbb{N}}
\newcommand{\C}{\mathbb{C}}
\newcommand{\R}{\mathbb{R}}

\newcommand{\TrSet}[0]{\mathcal{S}_{\mathrm{tr}}}

\newcommand{\Real}[0]{\mathbb{R}}
\newcommand{\TrQSet}[0]{\mathcal{S}_{\mathrm{Qtr}}}

\begin{document}
\newtheorem{notation}{Notation}[section]
\newtheorem{theorem}{Theorem}[section]
\newtheorem{lemma}{Lemma}[section]
\newtheorem{claim}{\rm{\it Claim}}[section]
\newtheorem{proposition}{Proposition}[section]
\newtheorem{corollary}{Corollary}[section]
\newtheorem{definition}{Definition}[section]
\newtheorem{example}{Example}[section]
\newtheorem{remark}{Remark}[section]
\newtheorem{problem}{Problem}[section]
\numberwithin{equation}{section}

\newenvironment{proofof}[2][Proof]%
{\begin{proof}[#1]%
\renewcommand{\qed}{\rule{0mm}{1mm}\hfill{\footnotesize\textsc{q.e.d.}\
(#2)}}}
{\end{proof}
}


   \title[Quantum state discrimination for supervised classification]
 {Quantum state discrimination for supervised classification 
   }

\author[Giuntini]{Roberto Giuntini}
\address[R.~Giuntini]{University of Cagliari,
    Via Is Mirrionis 1, I-09123 Cagliari, Italy.}
    \email{giuntini@unica.it}

\author[Freytes]{Hector Freytes}
\address[H.~Freytes]{University of Cagliari,
    Via Is Mirrionis 1, I-09123 Cagliari, Italy.}
    \email{hfreytes@gmail.com}
    
  \author[Park]{Daniel K. Park}
 \address[D.K.~Park]{Sungkyunkwan University Advanced Institute of Nanotechnology, Suwon, Korea.}
    \email{dkp.quantum@gmail.com}
    
      \author[Blank]{Carsten Blank}
      \address[C.~Blank]{data cybernetics, Landsberg, Germany}
    \email{blank@data-cybernetics.com}
      
       \author[Holik]{Federico Holik}
       \address[F.~Holik]{University of La Plata, Physics Institute}
    \email{olentiev2@gmail.com}

    \author[Chow]{Keng Loon Chow}
\address[K.L.~Chow]{University of Cagliari,
    Via Is Mirrionis 1, I-09123 Cagliari, Italy.}
    \email{leo.chow11@gmail.com}
    
    \author[Sergioli]{Giuseppe Sergioli}
\address[G.~Sergioli]{University of Cagliari,
    Via Is Mirrionis 1, I-09123 Cagliari, Italy.}
    \email{giuseppe.sergioli@gmail.com}


\maketitle
\begin{abstract}

In this paper we investigate the connection between quantum information theory and machine learning. In particular, we show how quantum state discrimination can represent a useful tool to address the standard classification problem in machine learning. Previous studies have shown that the optimal quantum measurement theory developed in the context of quantum information theory and quantum communication can inspire a new binary classification algorithm that can achieve higher inference accuracy for various datasets. Here we propose a model for arbitrary  multiclass classification inspired by quantum state discrimination, which is enabled by encoding the data in the space of linear operators on a Hilbert space. While our algorithm is quantum-inspired, it can be implemented on classical hardware, thereby permitting immediate applications.

\end{abstract}




\tableofcontents

\section{Introduction}

Quantum theory constitutes a new paradigm for information processing and provides unconventional ways to address computational problems. Advances in quantum computing have led to the development of algorithms that utilize quantum hardware to solve certain problems dramatically faster than any foreseeable classical hardware~\cite{10.2307/2899535,zalka1998simulating,shor1999polynomial,PhysRevLett.103.150502_HHL}. A commercially relevant family of problems for which the application of quantum algorithm promises certain computational benefits are found in the domain of machine learning. This gave birth to the new discipline known as quantum machine learning (QML). Several quantum machine learning algorithms have been proposed with clear quantum advantages~\cite{PhysRevLett.113.130503_qSQVM,qPCA,PhysRevA.94.022342,PhysRevA.97.042315,blank_quantum_2020,PARK2020126422}. However, the practical application of these algorithms are limited by the development of quantum hardware, which remains a long-term prospect.

Advances in quantum computing have led to another intriguing stream of research which aims to develop new classical algorithms inspired by quantum information processing to outperform existing methods, namely quantum-inspired classical algorithms~\cite{10.1145/3313276.3316310,tang2019quantuminspired,Arrazola2020quantuminspired}. The implications of this approach is significant not only for the domain of computational complexity theory, but also for practical applications. The complexity of implementing quantum machine learning algorithms on a quantum hardware and the emergence of quantum-inspired classical algorithms motivate the development of quantum-inspired machine learning (QIML)~\cite{holik_pattern_2018}. In principle, QIML deals only with ``mathematically quantum objects'': objects that are formally represented by different elements of the quantum formalism (such as density operators, also known as density matrices), but are not necessarily connected to actual quantum systems. Thus, the information stored in those objects can be formally managed by a classical computer. Recent findings show that the well-developed field of quantum state discrimination in quantum information theory and quantum communication can inspire new pattern recognition algorithms that can improve the binary classification accuracy of existing methods~\cite{Plos,IJTP,IJQI,sergioli_quantum-inspired_2021}.

In this work, we first explain the connection between quantum state discrimination and the quantum-inspired binary classification cited above, and then propose a quantum-inspired supervised machine learning algorithm for arbitrary multiclass classification. Our algorithm is based on using the mathematical framework of quantum mechanics to represent data and a quantum state discrimination technique known as \textit{Pretty Good measurement}. 
We show the theoretical derivation of this measurement strategy in the context of multiclass classification tasks in machine learning. We also show the classification accuracy of this quantum-inspired multiclass classifier can be improved by increasing the number of copies of the quantum object that encodes the data, at the cost of increasing the computational time. Since our approach does not require quantum hardware, it can be immediately implemented on existing classical hardware.

The paper is organized as follows. In Section 1, we provide a brief introduction regarding the general setting for supervised classification. In Section 2, the idea of quantum-inspired classifiers are described in detail. Section 3 is devoted to show the connection between quantum state discrimination and supervised classification; in particular, we provide some significant comments about quantum-inspired binary classification and then we introduce a more general quantum-inspired multiclass classification. Some final comments close the paper. 

\section{General setting for supervised classification}\label{General}

Supervised classification is one of the most important branches in machine learning~\cite{hastie_elements_2009}. It essentially consists of designing algorithms which learns by example in order to
classify 
objects. The term {\em supervised} refers to the intuitive idea that the entire process is supervised by an ``expert'' who first builds up a preliminary set of correctly classified objects and then,  on the basis of this dataset (called the {\em training dataset}), 
an algorithm is applied which would then allow one to classify new ``unseen'' objects (or objects from the training dataset) as accurately as possible.

Objects are described by a sequence number of $d$ {\em features} 
considered to be sufficiently relevant to characterize the objects in
the classification framework.  More formally, any object $x$ is associated to a vector 
$\vec{x}$ (called {\em object-vector} or
{\em feature-vector}) of a $d$-dimensional Hilbert space $\mathcal H^d$
\footnote{Unlike the standard presentations in machine learning,  we do not exclude features which may 
be represented as complex numbers. The feature-vector can be the raw data itself, or can be obtained via feature mapping from a lower dimensional space.}.
We define a pattern as a pair 
$$
(\vec{x_j}, \lambda_j)
$$ 
where $\vec{x_j}$ is a feature-vector and   
$\lambda_j$ is the class label which denotes the class with which the object is supposed to belong to.
For simplicity, we identify the set $L$ of all class labels with a finite 
sequence $(1,\ldots,\ell)$ of natural numbers that are in 
one-to-one correspondence with
the $\ell$ classes which the objects belong to. 
Thus, a training dataset can be represented as a set
$$
\TrSet:=\{(\vec x_1 ,\lambda_1),\ldots\,(\vec x_m,\lambda_m)\}.
$$
where $\lambda_j\in L$, $\forall j\in\{1,\ldots,m\}$.
Given any class label $i\in L$, we can define the set $\TrSet^i$ 
of all object-vectors whose associated class label is $i$:
\begin{equation}\label{stri}
\TrSet^i=\{\vec x_j \in\TrSet : \lambda_j=i\}.
\end{equation}
The cardinality of $\TrSet^i$ is denoted by 
$|\TrSet^i|$. Clearly, $\sum_{i=1}^{\ell} |\TrSet^i|=m$.

The task of supervised classification is to 
infer 
a function 
from the training dataset
$\TrSet$ using an algorithm (a classifier) 
which would classify
objects into one of the classes that the object is supposed to belong to, thus assigning a class label to an object-vector $\vec x$ as accurately as possible. 

Formally, a classifier can be defined as a map
$$
Cl: \C^d\,\to\,L. 
$$
There are several well known classification approaches in which different classifiers employ to classify
objects, such as distance-based or probability-based classification approaches.
In this paper, we present the following probability-based classification approach. Given a training dataset
$$ \TrSet := \{(\vec x_1,\lambda_1),\cdots,(\vec{x}_m,\lambda_m)\}, $$
one defines a map that associates to any feature-vector $\vec x$ a sequence
of $\ell$-numbers  belonging to the unit real-interval $[0,1]\subset \R$
$$ f:\,\C^d\,\to [0,1]^\ell. $$
The $i^{th}$-component of $f(\vec{x})$  will be denoted by $f(\vec x)_i$.

The meaning of  $f(\vec x)_i$ depends on the function $f$. 
For example, if $f(\vec x)$ is assumed to be a probability-vector
(i.e. $\sum_{i=1}^\ell f(\vec x)_i=1$),
  the  value $f(\vec x)_i$ can be interpreted as the probability
that the object $x$ (with associated feature-vector $\vec x$) belongs to the class labelled by $i$.

The classifier determined by $f$ (or simply, the $f$-classifier) is the map 
$$ Cl_f\,:\C^d\,\to L $$ 
that assigns to any feature-vector $\vec x\in\C^d$
the class label that is associated to the  greatest value of 
$f(\vec x)_i$, with $1\le i\le\ell$.
In other words, 
$$
Cl_f(\vec x)=\argmax_i \{f(\vec x)_i\,:\,1\le i\le\ell\}.
$$


Since it may happen that $f$ returns more than one class label when there are matching $f(\vec x)_i$ values,
we pose by convention

\begin{equation}\label{classifier}
Cl_f(\vec x):=\min\left\lbrace i\in L\, : \,f(\vec x)_i=\max_k\left\lbrace f(\vec x)_k\,,\,1\le k\le\ell \right\rbrace\right\rbrace.
\end{equation}

A classifier $Cl_f$ is called probabilistic iff 
$\sum_{i=1}^\ell f(\vec x)_i=1\; \forall \vec x\in\C^d.$
In other words, a classifier is probabilistic iff the 
sequence $\left( f(\vec x)_1,\ldots,f(\vec x)_\ell\right)$ is a probability-vector for any $\vec x$.

\bigskip


\section{Quantum classifiers}\label{Quantum Classifier}

Since object-vectors are defined in the Hilbert space, it is natural to treat it as a quantum state and apply mathematical techniques from quantum information theory. The underlying idea of this work is to encode object-vectors as density operators, and construct a classifier based on the quantum measurement technique developed in the realm of quantum information theory to optimally discriminate quantum states. In the following, we describe our encoding strategy for representing the set of object-vectors as density operators.

\subsection{Basic framework}\leavevmode

Given a training dataset, the construction of a quantum classifier is based on three fundamental steps: 
i) obtain a quantum feature map (or encoding) to encode the object-vectors of the training dataset into quantum objects (see Appendix~\ref{sec:appendixA});  
ii) find an appropriate function $f$ that determines the quantum classifier;
iii) apply the quantum classifier to the quantum encoded {\em object-vectors} to obtain class labels for the original (classical) objects.

Let us consider a training dataset
$\TrSet := \{(\vec x_1,\lambda_1),\cdots,(\vec x_m,\lambda_m)\}$. A quantum feature map (or encoding) is a map that associates
to any object-vector $\vec x$ of $\C^d$ a pure quantum state 
(called {\em object quantum-state}) $\rho_{\vec{x}}$ of a Hilbert space $\C^n$, whose dimension 
$n$ depends on the number of $d$ features. An example of quantum encoding is described in Appendix~\ref{sec:appendixA} (also see \cite{IJQI}). Without loss of generality, we use the density operator denoted by $\rho_{\vec{x}}$ rather than the state vector formalism to describe quantum states. Note that the density operator formalism follows naturally when we describe multiple object-vectors as mixed states in the following. Given a quantum encoding 
$\vec{x}\mapsto \rho_{\vec{x}}$,
a {\em quantum pattern} is any 
pair 
$$
(\rho_{\vec{x_j}},\lambda_j).
$$
A {\em quantum training dataset} is defined as the set of all quantum patterns
$$ \TrQSet = \left\{ (\rho_{\vec x_1}, \lambda_1), \ldots , (\rho_{\vec x_m}, \lambda_m) \right\} .$$

Given any class label $i\in L$, we can also define the set $\TrQSet^i$ as the set of all 
object quantum-states $\rho_{\vec x_j}$ that are associated to the set $\TrSet^i$ of all $i$-objects vectors:
\begin{equation}\label{isqtr}
\TrQSet^i=\{\rho_{\vec x_j}\, : \,\vec x_j\in \TrSet^i\}.
\end{equation} 

\begin{definition}\label{centroid}
Let $i\in L$ be a class label. The quantum centroid associated to $i$ denoted by
$\rho_{(i)}$
is the uniformly weighted convex combination of all $i$-object quantum-states

$$ \rho_{(i)}=\frac{1}{|\TrQSet^i|}\sum_{\vec x_j\in\TrSet^i}\rho_{\vec x_j}, $$
\end{definition}
where $|\TrQSet^i|$ is the cardinality of $\TrQSet^i$ (which is equivalent to $|\TrSet^i|$, the cardinality of $\TrSet^i$).


Thus, the $\ell$ class labels are in one-to-one correspondence with the
set $\{\rho_{(1)},\ldots,\rho_{(\ell)}\}$ for all quantum centroids.

\begin{remark}
Let us consider the classical centroid $\vec{C}_{(i)}=\frac{1}{|\TrSet^i|}\sum_{\vec x_j\in\TrSet^i}\vec x_j$ and let $\rho_{\vec C_{(i)}}$ the quantum encoding of the classical centroid $\vec C_{(i)}$. In general we can see that $\rho_{\vec C_{(i)}}\neq \rho_{(i)}$. In other words, the quantum centroid defined in Def.~\ref{centroid} does not correspond to the quantum encoding of the classical centroid.
\end{remark}
\begin{remark}
Let us consider a $d$-dimensional vector $\vec\delta$. In the classical case, the translation of the object-vector $\vec x_j\mapsto\vec x_j+\vec{\delta}$ results in the translation of the classical centroid as $\vec C_{(i)}\mapsto\vec C_{(i)}+\vec{\delta}$. However, the quantum centroid $\rho_{(i)}$ related to the object quantum-states $\{\rho_{\vec x_j}:\vec x_j\in \TrSet^i\}$ turns out to be non-translational invariant. This characteristic is shown to be beneficial for classification tasks ~\cite{IJTP,SS}.
\end{remark}

Let $
\TrSet:=\{(\vec x_1 ,\lambda_1),\ldots\,(\vec x_m,\lambda_m)\}$
be a training dataset with class labels
$L:=\{1,\ldots,\ell\}$.

Let $\mathcal{B}(\mathcal H)^+$ be the set of bounded and positive semidefinite operators acting on $\mathcal{H}$. A {\em quantum classifier\/} is a  classifier $Cl_f$, where the function $f$ is determined by a measurement 
$\mathcal M:\,L\to\,\mathcal{B}(\mathcal H)^+$
(see \cite{watrous_2018}, Def. 2.34).  

\begin{comment}
(\textcolor{red}{If we keep $L$ in the definition, we should explain that $L$ is the power set of $\mathcal{O}$ in the finite case, and the Borel sets in the infinite dimensional case; otherwise, we could just put simply $\mathcal{O}$, siplifying the definition a little bit. It depends on how we want to define it.}) 
\end{comment}

Recall that a measurement is defined as a map $\mathcal M$ from a finite non-empty set  $\mathcal O$ representing a set of possible outcomes of a physical quantity into $\mathcal{B}(\mathcal H)^+$, such that $\sum_{i\in\mathcal{O}}\mathcal{O}(i)=\mathbb{I}$.
In the case of von Neumann measurements, every $\mathcal{M}(i)$ is a projection.

We can now define formally the notion of a quantum classifier.
\begin{definition}\label{qclassifier}
A quantum classifier is a classifier $Cl_f$ (see Eq.(\ref{classifier}))
for which the map $f:\,\C^d\,\to [0,1]^\ell$ satisfies the following condition. There exists a measurement 
$\mathcal M:\,L\to\,\mathcal{B}(\mathcal H)^+$
such that 
$$ \forall\vec x\in\C^d: \,f(\vec x)_i=\Tr(\mathcal M(i)\rho_{\vec x}) $$
where tr is the trace of a matrix.
\end{definition}
By definition of measurement, 
$\sum_{i=1}^{\ell}\Tr(\mathcal M(i)\rho_{\vec x})=1$. Consequently, a quantum classifier is always probabilistic.\\

\subsection{Encoding in a larger Hilbert space}\leavevmode

An interesting question is whether classification accuracy can be improved by increasing the dimension of the state space of density matrices that represents feature-vectors. Since mapping a raw dataset in a low-dimensional feature space to a higher-dimensional feature space is a standard technique in the kernel method, such question arises naturally. Although the computation in the larger feature space increases the run time, the accurate prediction has more value in certain machine learning applications, such as those specialized in medical diagnosis. In particular, we consider encoding object-vectors $\vec{x}$ as the tensor product of $n$ copies of object quantum-states $\rho_{\vec{x}}$, i.e. $\vec{x}\mapsto \rho_{\vec{x}}\otimes\stackrel{n}{\ldots} \otimes\rho_{\vec{x}}$. In this case, we generalize the set defined in Eq. (\ref{isqtr}) as
$$ \TrQSet^{i^{(n)}}=\{\rho_{\vec x_j}\otimes\stackrel{n}{\ldots} \otimes\rho_{\vec x_j} : \,\vec x_j\in \TrSet^i \}. $$
Similarly, the $n$ copies-centroid can be defined as

\begin{equation}\label{ncentroid}
\rho_{(i)}^{(n)}=\frac{1}{|\TrQSet^i|}\sum_{\vec x_j\in\TrSet^i}\rho_{\vec x_j}\otimes\stackrel{n}{\ldots} \otimes\rho_{\vec x_j}.
\end{equation}
From the definition above, we should keep in mind that, in general, $\rho_{(i)}^{(n)}\neq\rho_{(i)}\otimes\stackrel{n}{\ldots} \otimes\rho_{(i)}$. With the construction above, the $\otimes^n$-generalization of a quantum classifier introduced in Def.~\ref{qclassifier} can be defined by a function $f:\mathbb C^{d^n}\rightarrow [0,1]^\ell$ and by a measurement 
$\mathcal M:\,L\to\,\mathcal{B}(\mathcal \otimes^n H)^+$. As we will show, this procedure turns out to be advantageous in improving classification accuracy.

\begin{comment}
{\color{blue}
DANIEL / LEO  

(\emph{The following concepts maybe could be formulated in a more suitable way and, maybe, also moved above (in the previous section)).}

 The second step is basically the core of the procedure and, like the standard scenario, it consists first to use part of the available dataset to set the classifier (training) and then to apply the classifier to the rest of the dataset (test set), i.e. to attribute to each object of the test set the class that it belongs to (testing). 
Finally, the last step is nothing but the inverse of the encoding. }
\end{comment}

\section{Classification inspired by quantum state discrimination}


Since 
feature-vectors
are defined as quantum states, we can conceptualize the machine learning task as finding an optimal quantum measurement for the feature-vectors, with respect to the 
quantum centroids. This permits the development of the new classification procedure inspired by quantum state discrimination~\cite{Helstrom1969,Barnett:09,Bae_2015}, which is a well-established field in the theory of quantum information~\cite{watrous_2018}. To make the connection more apparent, we briefly review quantum state discrimination as follows.

\subsection{Quantum state discrimination}\leavevmode

\begin{comment}
Consider an ensemble of quantum systems represented by a finite Hilbert space
$\C^d$, whose quantum states and associated probabilities are given by the set of pairs
$$
R=\{(\rho_1,\Prob_1),\cdots,(\rho_m,\Prob_m)\}, 
$$

\noindent where, for any $i$ $(1\le i\le m)$, $\rho_i$ is a density operator of $\C^d$,
and $(\Prob_i,\ldots,\Prob_m)$ is a probability-vector (namely a 
sequence of positive real numbers such that $\sum_{i=1}^m\Prob_i=1$).
The components of the vector $(\Prob_1,\cdots,\Prob_m)$ are called a priori probabilities. 

Suppose that Alice wishes to communicate an information to Bob by using a 
quantum system in one of the states of $R$. First of all, Alice tells Bob what 
the possible states are, together with their corresponding a priori probabilities.  
The a priori probabilities is related to the relative frequencies with which the states $ \rho_1,\ldots\,\rho_m$ 
are prepared. Alice selects a quantum system $\rho_i$ and hands it to Bob.
Bob has full knowledge of the relative frequencies in the ensemble $R$, but he does not 
know the actual state of the system he will receive from Alice in one run of the experiment. 
Bob's task is to determine the actual state of the system received by performing a single measurement over 
some suitably chosen physical observable.
\end{comment}

The task of discriminating quantum states is a fundamental problem in quantum information theory, with deep implications in quantum cryptography and quantum error correction~\cite{BENNETT20147}. The problem of quantum state discrimination can be summarized as follows.
Let us suppose that Alice wishes to send a message to Bob by using a quantum channel. To do this, Alice selects a state $\rho_i$ with an a priori probability $p_i$ from a given set of possible states in the Hilbert space $\mathbb C^d $. We indicate the set $R$ of these possible states with their respective a priori probabilities as follows 
$$
R=\{(\Prob_1,\rho_1),\cdots,(\Prob_{\ell},\rho_{\ell})\}, 
$$
where $(\Prob_i,\ldots,\Prob_{\ell})$ is a probability-vector (a 
sequence of positive real numbers such that $\sum_{i=1}^{\ell}\Prob_i=1$).

Bob knows a priori both the set of possible states and the associated probabilities. His task is to determine, by means of a suitably chosen measurement, the state $\rho_i$ he receives from Alice, and hence, the intended message. But the problem of finding an optimal strategy for discrimination among arbitrary states, has no known solution in the general case. While optimal solutions can be found for some particular cases (as, for example, when $R$ contains only mutually orthogonal states), in general, errors will necessarily occur in the discrimination process (see for example \cite{Bae_2015}).
This means that, in general, there exists no measurement $\mathcal M$ 
such that $\Tr(\mathcal M(i)\rho_j)=0$ when $i\neq j$, 
$\forall i,j$ $(1\leq i,j \leq {\ell})$. 
Thus, once Alice sends the $i$-th state to Bob, he can either conclude (erroneously) that he was given the state $\rho_j$ $(i\neq j)$ or, conversely, he can conclude (correctly) that he was given the state $\rho_j$ $(i=j)$ 
(successful discrimination). The average probability 
for Bob to perform a successful 
discrimination by means of a given measurement 
$\mathcal M$ is given by 
\begin{equation}\label{Measure}
\Prob _{succ}^{\mathcal M}(R)= \sum_{i=1}^{\ell}\Prob_i \Tr(\mathcal M(i)\rho_i).
\end{equation} 

In order to minimize the error probability in the discrimination problem, it is necessary to find an optimal measurement $\mathcal M$ which maximizes Eq.~(\ref{Measure}). Or equivalently, that 
minimizes the {\em discrimination error probability} 
$Err(R):=1-\Prob _{succ}^{\mathcal M}(R)$. One can prove (see \cite{watrous_2018}) that, for any ensemble $R$, there exists an {\em optimal measurement} $\mathcal M$ (shortened to $Opt(R)$), such that 
$$ Opt(R)= \max_{\mathcal N}\{ \Prob _{succ}^{\mathcal N}(R) \,: \mathcal{N}\,\,\text{is a measurement}\}. $$

\subsection{Helstrom measurement and binary classification}\leavevmode

In 1969, Helstrom reported an exact analytical description for the optimal measurement for ensembles of two quantum states
\cite{Helstrom1969}.
Let 
$$
R=\{(\Prob_1,\rho_1),(\Prob_2,\rho_2)\}
$$
be an ensemble of two quantum states 
with a priori probabilities $\Prob_1$ and $\Prob_2=1-\Prob_1$. 
Let us define the \emph{Helstrom observable} as 
\begin{equation}\label{HO}
    \Lambda=\Prob_1\rho_1-\Prob_2\rho_2.
\end{equation}
 Let $l_+$ and $l_-$ be the sets of all the eigenvectors determined by the positive and negative eigenvalues of $\Lambda$, respectively. 
Let $P_+=\sum_{\lambda\in l_+}P_{\lambda}$ and 
$P_-=\sum_{\lambda\in l_-}P_{\lambda}$, 
where $P_{\lambda}$ indicates the projection associated to the eigenspace determined by the eigenvalue $\lambda$. 
Intuitively, $P_+$ and $P_-$ represent the property of the measurement to correctly identify a state as being in state $\rho_1$ or $\rho_2$, respectively.
The set $\{P_+,P_-\}$ determines a von Neumann measurement, given that $P_++P_-=\mathbb I$
(see Appendix~\ref{sec:appendixB}). Helstrom proved that this measurement is optimal and the probability to successfully discriminate correctly between the two states has an upper bound \cite{Helstrom1969} -- called \emph{Helstrom bound} $(\tt H_b)$
-- given by
\begin{equation}\label{HB}
    {\tt H_b}(R)=\Prob_1 \Tr(P_+\rho_1)+\Prob_2 \Tr(P_-\rho_2).
\end{equation}

It turns out (see \cite{watrous_2018}) that 
$$
{\tt H_b}(R)=\frac{1}{2}+\frac{1}{2}{\tt T}\left(\Prob_1\rho_1,\Prob_2\rho_2\right)
$$
where ${\tt T}$ is the trace distance induced by the trace norm $\left\Vert \,^.\,\right\Vert_1$,
$$
{\tt T}\left(\Prob_1\rho_1,\Prob_2\rho_2\right):=\frac{1}{2}\left\Vert \Prob_1\rho_1-\Prob_2\rho_2\right\Vert_1
$$
In the case when $\Prob_1=\Prob_2=\frac {1}{2}$, we will simply write
$\tt H_b(\rho_1,\rho_2)$ instead of ${\tt H_b}(R)$.
$\tt H_b(\rho_1,\rho_2)$ satisfies the following properties (see \cite{watrous_2018}): 
\begin{enumerate}
\item[i)]\space  $0\le \tt H_b(\rho_1,\rho_2)\le 1$;
\item [ii)]$\tt H_b(\rho_1,\rho_2)=\frac{1}{2}$ iff $\rho_1=\rho_2$;
\item [iii)] $\tt H_b(\rho_1,\rho_2)=1$ iff $im(\rho_1)$ is orthogonal
to $im(\rho_2)$, where $im(\rho)$ is the subspace spanned by the image of $\rho$.
\end{enumerate}

Under Helstrom's formalism, let us now consider the particular case for binary classification, where $L:=\{1,2\}$.
After quantum encoding, we obtain the quantum training datasets $S^1_{Qtr}$ and $S^2_{Qtr}$, defined in accordance with Eq.~(\ref{isqtr}), and $\rho_1$ and $\rho_2$ as the respective quantum centroids obtained from Def.~\ref{centroid}.
In this way, it is possible to define the \emph{Helstrom observable} (as in Eq.~(\ref{HO})) for the two quantum centroids $\rho_1$ and $\rho_2$, where $p_1$ and $p_2$ are the a priori probabilities given by $p_1=\frac{|S^1_{Qtr}|}{|S_{Qtr}|}$ and $p_2=\frac{|S^2_{Qtr}|}{|S_{Qtr}|}$, respectively.
In this case, the quantum classifier introduced in Def.~\ref{qclassifier} for an arbitrary new ``unseen'' object quantum-state $\rho_{\vec x}$ (or from the quantum training dataset) is given by

$$Cl_f(\vec x):=\begin{cases}
    1, \,\, \text{if}\,\,\,\, \Tr(P_+ \rho_{\vec x})\geq \Tr(P_- \rho_{\vec x});\\
    2, \,\,\,\, \text{otherwise}.
    \end{cases}$$

An experiment is shown in \cite{Plos}, where the Helstrom Quantum Centroid (HQC) classifier is applied to fourteen different datasets and compared with eleven different standard classifiers, with the HQC classifier exhibiting a high level of accuracy performance when compared to these standard classifiers. In that paper, the application of the quantum state discrimination procedure to the problem of classification tasks
gave rise to the following notable observation: the higher the probability of successfully discriminating between quantum centroids, the better the accuracy performance of the quantum classifier. In other words, for the binary classification case, if one manages to increase the Helstrom bound (considered as a  measure of distinguishability between the two quantum centroids), then one would expect the accuracy of the quantum classifier to be higher. The experiment given in \cite{Plos} provides an empirical evidence of this intuition. Also as shown in Fig.~\ref{Isto}, for the given fourteen datasets that were investigated, we obtained a Pearson coefficient of $\approx 0.9468$ which indicates a strong positive (linear) relationship between the Helstrom bound and balanced accuracy score.
\begin{figure}[t]
\includegraphics[width=0.90\columnwidth]{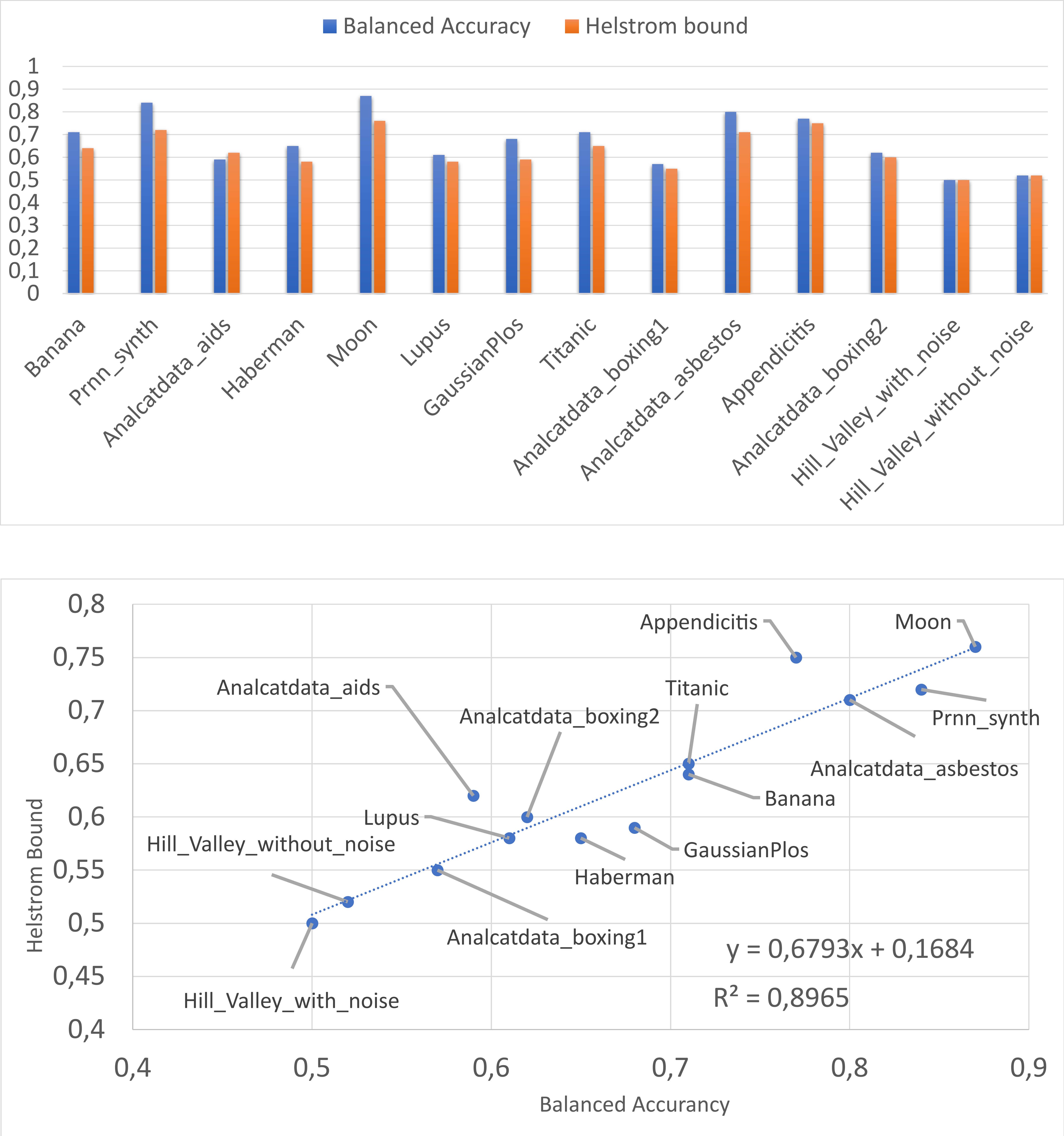}\caption{\label{Isto}\textbf{\small Balanced Accuracy vs. Helstrom Bound.} (Upper Fig.) \small Balanced accuracy and the Helstrom bound were calculated and compared for fourteen example datasets. (Lower Fig.) A Pearson coefficient of $0.9468$ was obtained indicating the Helstrom bound and balanced accuracy has a strong positive (linear) relationship. The regression line obtained is ${\tt H_b} \approx \frac{2}{3}$(balanced accuracy)$+\frac{1}{6}$ with $R^2 \approx 0.8965$, indicating a good fit of the linear regression model with respect to the data.}

\end{figure}

\subsubsection{Advantages of encoding in a larger Hilbert space}\leavevmode

As discussed at the end of Section \ref{Quantum Classifier}, it is possible to perform
the tensor product of $n$ copies of each object quantum-state. This allows us to define the new quantum centroids $\rho_1^{(n)}$ and $\rho_2^{(n)}$. Consequently, the respective Helstrom measurement is given by $\{P_+^{(n)},P_-^{(n)}\}$. In this way, the quantum classifier takes the form

$$Cl^{(n)}_f(\vec x):=\begin{cases}
    1, \,\, \text{if}\,\,\,\, \Tr(P_+^{(n)} \rho^{(n)}_{\vec x})\geq \Tr(P_-^{(n)} \rho^{(n)}_{\vec x});\\
    2, \,\,\,\, \text{otherwise}.
    \end{cases}$$

The empirical results obtained in \cite{Plos} strongly suggest that 
taking tensor products of artificial copies of the object quantum-states turns out to be beneficial for classification tasks. As a theoretical support to this claim, in this paper, we prove the following.

\begin{theorem}\label{th:morecopies}
For any $n\in\N^+$,
$$
{\tt H_b} (\rho_{(1)}^{(n)},\,\rho_{(2)}^{(n)})\leq {\tt H_b}(\rho_{(1)}^{(n+1)},\,
   \rho_{(2)}^{(n+1)}).
$$
\end{theorem}
\begin{proof}
See Appendix II.
\end{proof}

According to Theorem \ref{th:morecopies}, increasing the number of copies of the object quantum-states (and consequently, increasing the number of factors in the tensor product) gives a higher Helstrom bound (which means a higher success probability for discriminating the two quantum centroids). 
As a result, the accuracy performance of the quantum classifier increases. As an example, by applying the HQC classifier to six different datasets (in the experiment as detailed in \cite{Plos}), we show in Fig.~\ref{Four} how the application of this procedure improves the accuracy performance measured by three different statistical scores (balanced accuracy, f1-score and k-Cohen). 

\begin{figure}[t]
\centering
\includegraphics[width=0.99\columnwidth]{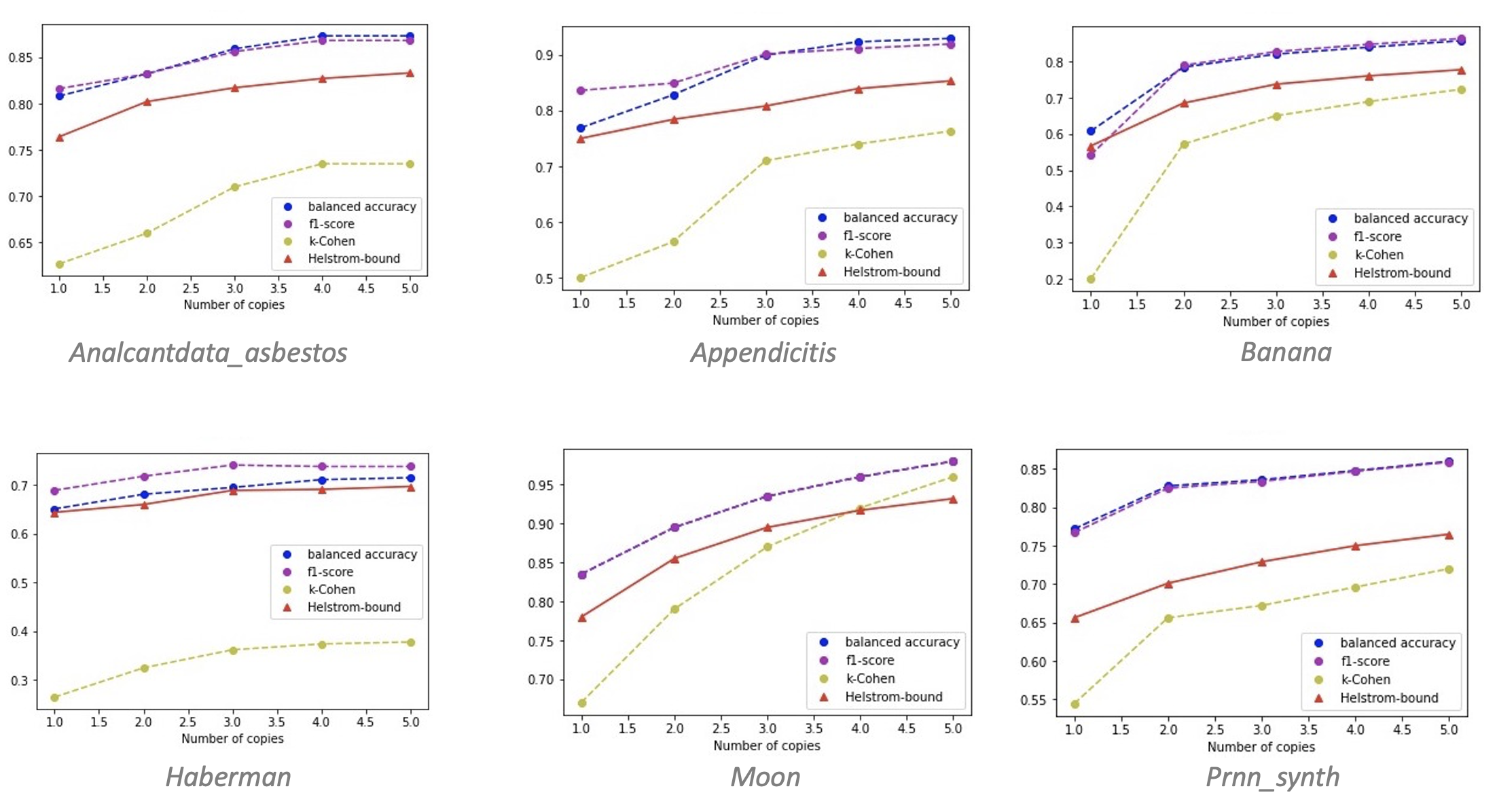}\caption{\label{Four} These plots show the improvement of the balanced accuracy, $f1$-score, k-Cohen and Helstrom bound, respectively, with the increasing of the number of the tensor copies of the original data.}
\end{figure}

Many questions remain open, and an extensive investigation into the relation between the Helstrom bound and the statistical scores in machine learning problems is still ongoing. The results obtained so far motivates further research into this.

It is also relevant to note that, even if increasing the number of copies in the tensor product, in principle, produces an improvement to the accuracy performance of the quantum classifier, at the same time it also produces an additional computational cost (i.e. run time of the algorithm). Hence, depending on the size of the dataset and on the machine being used, the maximum number of tensor copies should be chosen by a user based on practical considerations. This leads to another interesting future work that aims to utilize quantum hardware to speedup the classification protocol.

\subsection{Pretty Good measurement and multiclass classification}\leavevmode

A common strategy to perform multiclass classification is based on combinatorial compositions of binary classifications (i.e. One-vs-One or One-vs-Rest strategies).
One disadvantage with these strategies is that if a dataset contains a large number of classes, it would increase the computational time of the algorithm. In the following, we show how quantum state classification suggests an alternative method which allows us to avoid this combinatorial complexity for multiclass classification.

Given an ensemble of possible states with their respective a priori probabilities
\begin{equation}\label{ensamble}
R=\{(\Prob_1,\rho_1),\cdots,(\Prob_{\ell},\rho_{\ell})\},
\end{equation}
it may be difficult to find an analytical description for the exact
optimal measurement associated to $R$. One possible solution is to search for a sub-optimal measurement that can be expressed in an analytical form. This is known as the so-called {\em Pretty Good measurement}
\cite{watrous_2018}, which we will introduce in the following. 

Let us define the average state of $R$ as
$$
\sigma=\sum_{i=1}^{\ell}\Prob_i\rho_i
$$
For any $i:\,1\le i\le {\ell}$, let us define the following operator
$$
E_i=(\sigma\urcorner)^{1/2}\Prob_i\rho_i(\sigma\urcorner)^{1/2},
$$
where $\sigma\urcorner$ is the pseudoinverse (or Moore-Penrose inverse) of $\sigma$.
The operator $E_i$ is well defined. Indeed, by Theorem \ref{th:pseudo1} (iv) (see Appendix~\ref{sec:appendixC}),
$\sigma\urcorner$ is positive since $\sigma$ is positive and 
consequently, the square root of $\sigma\urcorner$ exists and is unique.
Since $\Prob_i\rho_i$ is positive semidefinite and it is enclosed on the left and on the right by a self-adjoint operator $(\sigma\urcorner)^{1/2}$, we conclude $E_i$ is also positive semidefinite (see \cite{Ho}, p.431). By Theorem  \ref{th:pseudo1} (vi) (see Appendix~\ref{sec:appendixC}), we have 
$\sigma\sigma\urcorner=\sigma\urcorner\sigma$. Consequently, by Theorem 6.6.4 (see \cite{Fri}, p.223), 
$(\sigma\urcorner)^{1/2}$ commutes with $\sigma$. 
As a result, we have
\begin{align}
\sum_{i=1}^{\ell} E_i&=\sum_{i=1}^{\ell} (\sigma\urcorner)^{1/2}\Prob_i\rho_i(\sigma\urcorner)^{1/2} \notag \\
                &=(\sigma\urcorner)^{1/2}\sigma (\sigma\urcorner)^{1/2} \notag \\
                &=\sigma(\sigma\urcorner)^{1/2}(\sigma\urcorner)^{1/2} \notag \\
                &=\sigma(\sigma\urcorner) \notag \\
                &=P_{im(\sigma)}, \label{E-im}
\end{align}
by Theorem \ref{th:pseudo1} (vi), where $P_{im(\sigma)}$ is the projection onto the subspace spanned by the image of $\sigma$.

From Theorem \ref{th:pseudo1} (vi), $\sigma \sigma\urcorner=P_{im(\sigma)}$ is a projection.
Since any $E_i$ is positve semidefinite, we have that $E_i\le \sum_{i=1}^{\ell} E_i=P_{im(\sigma)}$ is 
bounded by the identity operator $\mathbb{I}$
and therefore it is an effect operator: $\mathbb{O}\le E_i\le \mathbb{I}$.
However, the set $\{E_i\}_{i=1}^{\ell}$ does not determine a measurement since in general 
$\sum_{i=1}^{\ell} E_i<\mathbb{I}$. 
The set $\{E_i\}_{i=1}^{\ell}$ can easily be transformed into a set of ${\ell}$ effects that 
finally induce a measurement.

For any $i$ with $1\le i\le {\ell}$, let us define the following operators
\begin{equation}\label{F}
F_i=E_i +\frac{1}{{\ell}}P_{ker(\sigma)},
\end{equation}
where $P_{ker(\sigma)}$ is the projection associated to the subspace spanned by
the kernel of $\sigma$.
It turns out that the map
$$
\mathcal{F}:\,\{1,2,\cdots,{\ell}\}\,\to \mathcal{B}(\C^n)^+
$$ 
is a measurement since, as we have seen above, $\sum_{i=1}^{\ell} E_i=P_{im(\sigma)}$ and
consequently
\begin{align}
\sum_{i=1}^{\ell} F_i&=\sum_{i=1}^{\ell}(E_i+\frac{1}{{\ell}}P_{ker(\sigma)}) \notag \\
&=\sum_{i=1}^{\ell} E_i+\frac{1}{{\ell}}\sum_{i=1}^{\ell} P_{ker(\sigma)} \notag \\
&=P_{im(\sigma)}+P_{ker(\sigma)} \label{F-im-ker} \\
&=\mathbb{I}. \notag
\end{align}
We have now induced a measurement where $\sum_{i=1}^{\ell} F_i=\mathbb{I}$.
Since $\sigma$ is invertible iff $ P_{ker(\sigma)}=\mathbb{O}$, we can conclude, from Eqs. (\ref{E-im}) and (\ref{F-im-ker}), that $E_i=F_i$ iff $\sigma$ is invertible.

The map $\mathcal F$, called Pretty Good measurement \cite{Wooters_94}, 
is sub-optimal ~\cite{Barnum_Knill_2002} since
$$
\Prob _{succ}^{\mathcal F}(R)\,\ge \,Opt(R)^2.
$$

Let us now turn to the general problem of multiclass classification. We will follow a similar formalism to that of the Helstrom observable (for the binary classification case), but using a classifier based on the Pretty Good measurement formalism. After the quantum encoding procedure, we consider the quantum training datasets $S^i_{Qtr}$ as defined in Eq.~(\ref{isqtr}) and the respective quantum centroids $\rho_{(i)}$ as defined in Def.~\ref{centroid}. Hence, it is possible to consider the ensamble $R$ as defined in Eq.~(\ref{ensamble}) where $p_i=\frac{|S^i_{Qtr}|}{|S_{Qtr}|}$ (for any $i\in\{1,\ldots,{\ell}\}$). In this way, $R$ induces a set of operators $F_i$ as defined in Eq.~(\ref{F}). In this case, the multiclass quantum classifier introduced in Def.~\ref{qclassifier} for an arbitrary new ``unseen'' object quantum-state $\rho_{\vec x}$ (or from the quantum training dataset) is given by
\begin{equation}\label{PgmCl}
Cl_f(\vec x):=\min_i\{i\in\{1,\cdots,{\ell}\} : p_i \Tr(F_i\rho_{\vec{x}})=\max_k\{p_k \Tr(F_k\rho_{\vec{x}}), 1\le k\le {\ell}\}\}.
\end{equation}
Notice that if $Tr (P_{Ker(\sigma)})=0$, we can replace $F_i$ by $E_i$ in the above equation.

We can also generalize this framework by taking the tensor product of $n$ copies of states. Thus, the definition of the multiclass quantum classifier introduced in Eq.~(\ref{PgmCl}), can naturally be extended as 
\begin{align*}
Cl_f(\vec x) :=
& \min_i\{i\in\{1,\cdots,m\} : \\ 
& \,p_i \Tr(F^{(n)}_i\rho^{(n)}_{\vec{x}})=\max_k\{p_k \Tr(F^{(n)}_k\rho^{(n)}_{\vec{x}}), 1\le k\le {\ell}\}\}. \notag
\end{align*}
    
As in the Helstrom measurement case, it is possible to define a bound for the Pretty Good measurement as follows
$${\tt PGM_b}(R)=\sum_{i=1}^{\ell} p_i \Tr (F_i\rho_i).$$
    
It is currently not known whether it is possible to obtain a general result regarding the relation between the value of the PGM bound $(\tt PGM_b)$ and the number of tensor products (as in Theorem \ref{th:morecopies}) for the Pretty Good measurement classifier. We present an initial insight and empirical evidence on this, suggesting an analogous result, in Appendix~\ref{sec:appendixD}.



\section{Conclusion}
In this work we have established a connection between the problems of quantum state discrimination and classification tasks in the machine learning context. Specifically, we assert that the former can provide relevant benefits to the latter. The proposed strategy is based on the following steps. First we obtain a quantum feature map (or encoding) of real feature-vectors to pure quantum states, then to each class we associate a quantum centroid as a mixed quantum state. In other words, the feature space is defined by the space of linear operators on a Hilbert space. Next we develop a quantum version of the supervised classification task that has the general property that the performance of the classification process increases with the capability of distinguishing among the quantum centroids. We have presented theoretical and empirical evidences to support this. Finally, we show that in the quantum framework, making duplicate copies of the initial data turns out to be beneficial for the classification process. This idea is similar to increasing the dimension of the feature space that is encountered in the kernel trick. This strategy, previously provided only for binary classification, is hereby introduced in the general setting for multiclass classification developed using the Pretty Good measurement formalism.

\newpage
\begin{appendix}
\section{Quantum feature maps (encodings)}
\label{sec:appendixA}
There are, or course, infinitely many ways to map an object-vector into a density matrix. Different encodings have been already considered in previous works \cite{IJQI}. Here, as an instance, we focus on the following encoding procedure. For each $\vec x = \left[x^1, \ldots, x^d \right]\in \Real^d$, let 
$$\tilde{x}=[\frac{x^1}{\sqrt{\sum_{D=1}^d(x^D)^2+1}}, \ldots, \frac{x^d}{\sqrt{\sum_{D=1}^d(x^D)^2+1}},\frac{1}{\sqrt{\sum_{D=1}^d(x^D)^2+1}}]$$

\noindent Next, we give the following.

\begin{definition}
The amplitude encoding is the map $$\Real^d\ni\vec x\mapsto\rho_{\vec x}=\tilde{x}^\dagger\tilde{x}$$  
\end{definition}

Note that $\rho_{\vec x}$ is a pure state in $\mathcal D^{(d+1)}$. We say that $\rho_{\vec x}$ is the object quantum-state associated to object-vector $\vec x$.

\section{Helstrom bound and tensor copies}
\label{sec:appendixB}
In this Appendix we give the proof of Theorem \ref{th:morecopies}: \\

\noindent For any $k\in\N^+$,
\begin{equation}
{\tt H_b} (\rho_{(1)}^{(k)},\,\rho_{(2)}^{(k)})\leq {\tt H_b}(\rho_{(1)}^{(k+1)},\,
   \rho_{(2)}^{(k+1)}). \notag
\end{equation}
\begin{proof}
Since the Helstrom bound of two density operators $\rho_1$ and $\rho_2$
is given by $\frac{1}{2}+\frac{1}{2}{\tt T}(\Prob_1\rho_1,\Prob_2\rho_2)$ -- where $\tt T$ is the trace distance induced by the trace norm $\left\Vert \,^.\,\right\Vert_1$ --
in order to prove Theorem \ref{th:morecopies}, it suffices to show that
$$
{\tt T} (\rho_{(1)}^{(k)},\,\rho_{(2)}^{(k)})\leq {\tt T}(\rho_{(1)}^{(k+1)},\,
\rho_{(2)}^{(k+1)}) .
$$
Let us consider the Hilbert space $\mathcal H=\mathcal H_1\otimes\mathcal H_2$, where
$\mathcal H_1=\C^d$ and $\mathcal H_2=\otimes^k\C^d$. Let $\tt pT_1$ be the partial trace of the first
component of $\mathcal H$, i.e., $\mathcal H_1$. 

\begin{comment}

It is well know \cite{NC00} that
the partial trace function is a trace-preserving quantum operation. \textcolor{red}{I don't think that the preservation of the trace is behind the derivation of the following four equalities. It is just linearity of the partial trace. Intuitively, one has a $k+1$ product state; if one traces out the first system, the result is just the tensor product of the reminding $k$ systems.} 
\end{comment}

Using linearity of the partial trace operator
\begin{align}
{\tt pT_1}(\rho_{(1)}^{(k+1)}) &= 
                      {\tt pT_1}\left(\frac{1}{|\TrQSet^1|}\sum_{\vec x_j\in\TrSet^1}\otimes^{(k+1)}\rho_{\vec x_j}\right) \notag \\ 
     & =\frac{1}{|\TrQSet^1|}\sum_{\vec x_j\in\TrSet^1}{\tt pT_1}\left(\otimes^{(k+1)}\rho_{\vec x_j}\right) \notag \\ 
        &=\frac{1}{|\TrQSet^1|}\sum_{\vec x_j\in\TrSet^1}\otimes^{(k)}\rho_{\vec x_j} \notag \\ 
&=\rho_{(1)}^{(k)} \label{pTr1-1}
\end{align}

The above list of equations can be physically interpreted as follows: if one has a product state of $k+1$ components, tracing out the first system will not alter the rest of the $k$ states. Similarly, 
\begin{align}
{\tt pT_1}(\rho_{(2)}^{(k+1)})=
                       \frac{1}{|\TrQSet^2|}\sum_{\vec x_j\in\TrSet^2}\otimes^{(k)}\rho_{\vec x_j}
=\rho_{(2)}^{(k)}. \label{pTr1-2}
\end{align}
The trace distance satisfies a contractivity property under the action of 
complete trace preserving of positive maps (i.e., trace preserving quantum operators, see \cite{Rus}). Since $\tt pT_1$ is a trace-preserving quantum operation, we have that
\begin{align}
{\tt T} \left({\tt {pT}_1}\left(\rho_{(1)}^{(k+1)}\right),\,{\tt pT_1}\left(\rho_{(2)}^{(k+1)}\right)\right)
\leq {\tt T}\left(\rho_{(1)}^{(k+1)},\,\rho_{(2)}^{(k+1)}\right). \label{pTr1-1-2}
\end{align}
Thus, by 
Eqs. (\ref{pTr1-1}), (\ref{pTr1-2}) and (\ref{pTr1-1-2}) we can conclude that 
$${\tt T} (\rho_{(1)}^{(k)},\,\rho_{(2)}^{(k)})\leq {\tt T}(\rho_{(1)}^{(k+1)},\,
\rho_{(2)}^{(k+1)}) .
$$
\end{proof}

\section{Mathematical properties of the pseudoinverse}
\label{sec:appendixC}
Let $\mathcal{\C}^n$ be the finite complex Hilbert space of dimension $n$. 
$\mathcal{B}(\mathcal H)^+$ will denote the set of all positive semidefinite
bounded linear operator of 
$\mathcal{\C}^n$. The set of 
all projections of $\mathcal H$ will be denoted by 
$P(\mathcal H)$. Using $E(\mathcal H)$ we will denote the set of all
effects of $\mathcal H$, i.e. the set of all positive operator in 
$\mathcal{B}(\mathcal H)^+$ that are bounded by the identity operator
$\mathbb I$. Thus, an effect is a an operator $E$ of $\mathcal H$ such that
$\mathbb O\le E\le \mathbb I$.
\begin{definition}
Let $A$ be a linear operator of $\mathcal{\C}^n$.
The pseudoinverse (or Moore-Penrose inverse) of $A$ is an operator $X$ of $\mathcal{\C}^n$ such that the following conditions are satisfied:
\begin{enumerate}
\item[i)] \,\,$AXA=A$;
\item[ii)] \,\,$XAX=X$;
\item[iii)] \,\,$(AX)^\dagger=AX$, where $^\dagger$ is the adjoint operation;
\item[iv)] \,\,$(XA)^\dagger=XA$.
\end{enumerate}
\end{definition}
One can prove that the pseudoinverse of any operator exists and is unique.
The pseudoinverse of an operator $A$  will be denoted by $A\urcorner$.
It turns out that if $A$ is invertible, then the inverse of $A$ (i.e. $A^{-1}$)
coincides with $A\urcorner$.
\begin{theorem}{\label{th:pseudo1}}
Let $A$ be a linear operator of $\mathcal{\C}^n$.
The following properties hold:
\begin{enumerate}
\item[i)] \,\,$A$ is invertible iff $A^{-1}=A^\urcorner$;
\item[ii)] \,\,$(A^\urcorner)^\urcorner=A$;
\item[iii)] \,\,$(A^\dagger)^\urcorner=(A^\urcorner)^\dagger$;
 \item[iv)] \,\,if $A\in \mathcal{B}(\mathcal H)^+$, then 
               $A^\urcorner\in \mathcal{B}(\mathcal H)^+$;
\item[v)] \,\,$AA^\urcorner$ and $A^\urcorner A$ are projections;
\item[vi)] \,\,if $A\in \mathcal{B}(\mathcal H)^+$, then 
          $AA^\urcorner=A^\urcorner A=P_{im(A)}$, where 
         $P_{im(A)}$ is the projection that projects onto the image of $A$.
\end{enumerate}
\end{theorem}

\section{An example of the PGM bound and tensor copies.}
\label{sec:appendixD}
\begin{theorem}

Let us consider the special quantum datasets $ R=\{\rho_1,\rho_2\}$ and $S=\{\sigma_1,\sigma_2\}$ (representing 2 classes) where $\rho_i$ and $\sigma_i$ are diagonal matrices, i.e. 
$\rho_i=\begin{pmatrix}
 1-r_i & 0 \\
                       0 & r_i
                     \end{pmatrix}$
                     and
                     $\sigma_i=\begin{pmatrix}
 1-s_i & 0 \\
                       0 & s_i
                     \end{pmatrix}$, with $0\leq r_i,s_i \leq 1$.

                     Also lets consider the datasets $R^{(2)}=\{\rho_1\otimes\rho_1,\rho_2\otimes\rho_2\}$ and $S^{(2)}=\{\sigma_1\otimes\sigma_1,\sigma_2\otimes\sigma_2\}$.
                     
                     Let us calculate the quantum centroids for the sets $R$ and $S$:   $\rho=\frac{1}{2}(\rho_1+\rho_2)$ and $\sigma=\frac{1}{2}(\sigma_1+\sigma_2)$ and, analogously, for the sets $R^{(2)}$ and $S^{(2)}$: 
$\rho^{(2)}=\frac{1}{2}(\rho_1\otimes\rho_1+\rho_2\otimes\rho_2)$ and $\sigma^{(2)}=\frac{1}{2}(\sigma_1\otimes\sigma_1+\sigma_2\otimes\sigma_2).$
                     
Let ${\tt PGM}_b$ and ${\tt PGM}_b^{(2)}$, the Pretty Good measurement bound built over $\rho$ and $\sigma$ and the Pretty Good measurement bound built over $\rho^{(2)}$ and $\sigma^{(2)}$, respectively. 

We show that ${\tt PGM}_b<{\tt PGM}_b^{(2)}$ for any initial sets $R$ and $S$.
\end{theorem}
\begin{proof}
Following the standard procedure of the $PGM$ formalism, we first calculate the average state $$H=\frac{1}{2}(\rho + \sigma).$$

\noindent Then, we calculate:
$$G_{\rho}=\frac{1}{2}\sqrt{H^{\neg}}\rho\sqrt{H^{\neg}}+\frac{1}{2}Ker(H)$$ and 
$$G_{\sigma}=\frac{1}{2}\sqrt{H^{\neg}}\sigma\sqrt{H^{\neg}}+\frac{1}{2}Ker(H)$$

\noindent where $\neg$ represents the pseudoinverse function. Let us notice that, given the fact that $H$ has a diagonal form, then the pseudoinverse coincides with the inverse and consequently $Ker(H)=0$.

Now let us calculate ${\tt PGM}_b=\frac{1}{2}({\Tr}(G_{\rho}\rho)+{\Tr}(G_{\sigma}\sigma))$. By a straightforward calculation, it follows that: $${\tt PGM}_b=\frac{2((r_1+r_2)(s_1+s_2-1)-s_1-s_2)}{(r_1+r_2+s_1+s_2)(r_1+r_2+s_1+s_2-4)}.$$ 
Following the same strategy, it is possible to define $H^{(2)}$, $G_{\rho}^{(2)}$, $G_{\sigma}^{(2)}$ and ${\tt PGM}_b^{(2)}$. It can be easily shown that the difference between ${\tt PGM}_b^{(2)}$ and ${\tt PGM}_b$ can be written as:
${\tt PGM}_b^{(2)}-{\tt PGM}_b=\frac{1}{2}(A+B+C+D),$ where

$$A= \frac{(r_1(r_1-2)+r_2(r_2-2)+2)^2}{r_1(r_1-2)+r_2(r_2-2)+s_1(s_1-2)+s_2(s_2-2)+4};$$

$$B=\frac{2(r_1(r_1-1)+r_2(r_2-1))^2}{r_1(1-r_1)+r_2(1-r_2)+s_1(1-s_1)+s_2(1-s_2)};$$

$$C=\frac{(r_1+r_2)^2}{r_1^2+r_2^2+s_1^2+s_2^2};$$

$$D=\frac{4((r_1+r_2)(s_1+s_2)-r_1-r_2-s_1-s_2)}{(r_1+r_2+s_1+s_2)(4-r_1-r_2-s_1-s_2)}.$$

\bigskip

It trivially follows that $\forall (r_1,r_2,s_1,s_2)\in(0,1)$, $A$, $B$ and $C$ are non-negative and $D\leq 0$.

It can be seen that $Grad({\tt PGM}_b^{(2)}-{\tt PGM}_b)=0$ for $r_1=r_2=s_1=s_2\simeq 0.7931$ and for these values (which provides the minimum difference) ${\tt PGM}_b^{(2)}-{\tt PGM}_b=0$. Hence our claim holds.

\end{proof}

In an analogous way, it is possible to define ${\tt PGM}_b^{(3)}$, ${\tt PGM}_b^{(4)}$ and so on. We empirically study the effect of increasing the number of tensor copies (i.e. increasing the dimension of the feature space) on the PGM bound via numerical simulations based on $10^4$ random 2-feature vectors. The simulation results are shown in Fig.~\ref{fig:simulation}.

\begin{figure}[t]
    \centering
    \includegraphics[width=0.99\columnwidth]{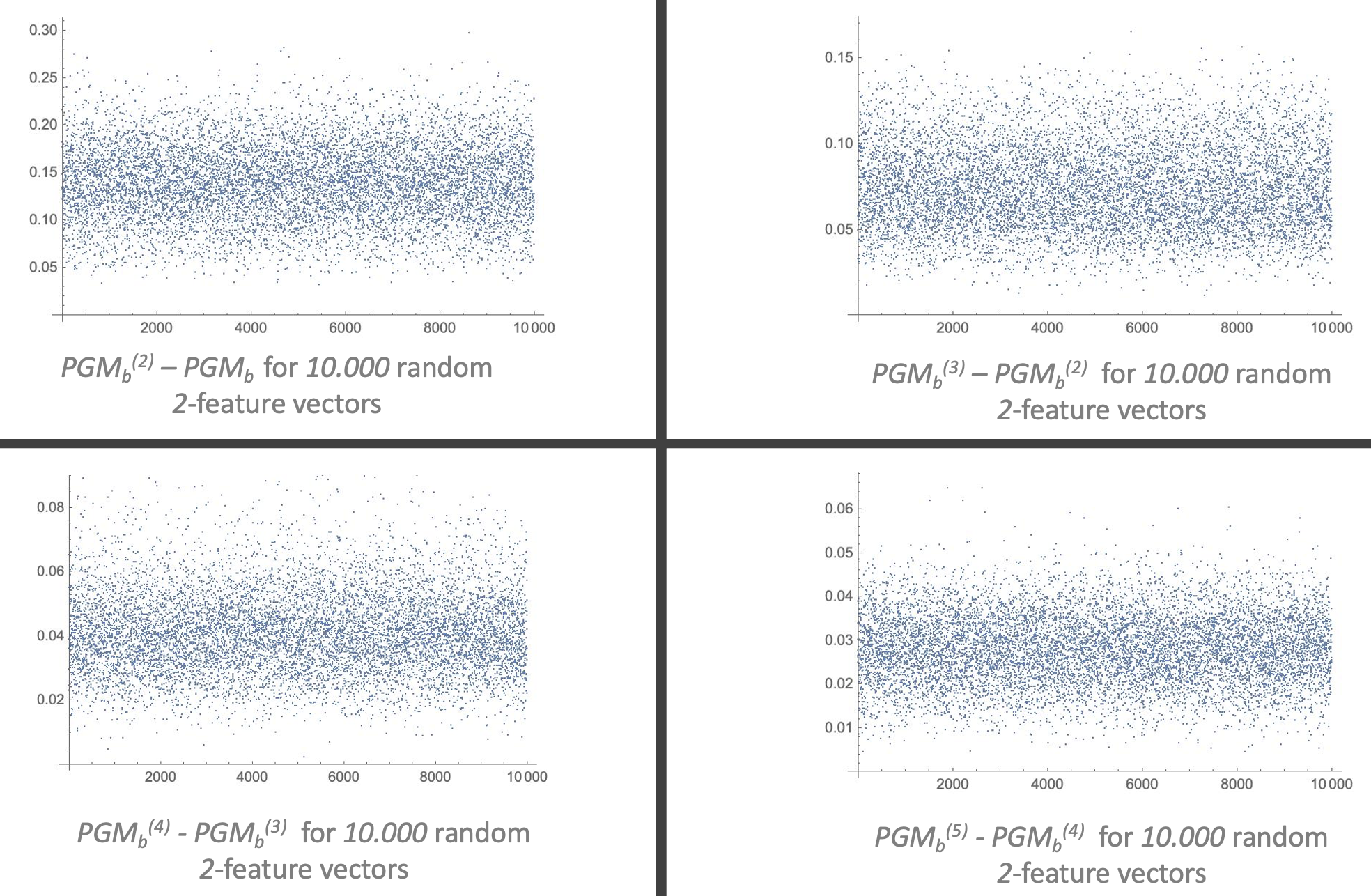}
    \caption{\label{fig:simulation} Here we show four simulations given by considering, each time, $10^4$ random 2-feature vectors and by calculating (for each case) the value of the quantity ${\tt PGM}_b^{(n+1)} -{\tt PGM}_b^{(n)}$, for $n\in\{1,4\}$. For all the instances it is verified that ${\tt PGM}_b^{(n+1)}\geq {\tt PGM}_b^{(n)}$.}
\end{figure}
%
\begin{comment}
 For the journal submission, we should perhaps also plot the mean (and the error bar could be determined from the standard deviation) for $PGM_b^{n}$ - $PGM_b$ as a function of $n$ to observe a scaling behaviour (similar to Figure 2).
\end{comment}
 
These results show how ${\tt PGM}_b^{(n+1)}\geq {\tt PGM}_b^{(n)}$. A formal and general proof of this observation remains a future work.

\end{appendix}

\section*{Acknowledgments}
D.K.P. gratefully acknowledges support from the National Research Foundation of Korea (No. 2019R1I1A1A01050161), and Quantum Computing Development Program (No. 2019M3E4A1080227).
R.G gratefully acknowledges support from the R.G. is grateful to the RAS (Regione Autonoma
della Sardegna) (project code: RASSR40341, A Semantic Extension of Quantum Computationsl Logic). G.S. is grateful to Fondazione di Sardegna (project code: F71I17000330002) and to the Prin project ``Logic and Cognition. Theory, experiments, and applications". 


\begin{thebibliography}{10}

\bibitem{10.2307/2899535}
Seth Lloyd.
\newblock Universal quantum simulators.
\newblock {\em Science}, 273(5278):1073--1078, 1996.

\bibitem{zalka1998simulating}
Christof Zalka.
\newblock Simulating quantum systems on a quantum computer.
\newblock {\em Proceedings of the Royal Society of London. Series A:
  Mathematical, Physical and Engineering Sciences}, 454(1969):313--322, 1998.

\bibitem{shor1999polynomial}
Peter~W Shor.
\newblock Polynomial-time algorithms for prime factorization and discrete
  logarithms on a quantum computer.
\newblock {\em SIAM review}, 41(2):303--332, 1999.

\bibitem{PhysRevLett.103.150502_HHL}
Aram~W. Harrow, Avinatan Hassidim, and Seth Lloyd.
\newblock Quantum algorithm for linear systems of equations.
\newblock {\em Phys. Rev. Lett.}, 103:150502, Oct 2009.

\bibitem{PhysRevLett.113.130503_qSQVM}
Patrick Rebentrost, Masoud Mohseni, and Seth Lloyd.
\newblock Quantum support vector machine for big data classification.
\newblock {\em Phys. Rev. Lett.}, 113:130503, Sep 2014.

\bibitem{qPCA}
Seth Lloyd, Masoud Mohseni, and Patrick Rebentrost.
\newblock Quantum principal component analysis.
\newblock {\em Nature Physics}, 10(9):631--633, 2014.

\bibitem{PhysRevA.94.022342}
Maria Schuld, Ilya Sinayskiy, and Francesco Petruccione.
\newblock Prediction by linear regression on a quantum computer.
\newblock {\em Phys. Rev. A}, 94:022342, Aug 2016.

\bibitem{PhysRevA.97.042315}
Nana Liu and Patrick Rebentrost.
\newblock Quantum machine learning for quantum anomaly detection.
\newblock {\em Phys. Rev. A}, 97:042315, Apr 2018.

\bibitem{blank_quantum_2020}
Carsten Blank, Daniel~K. Park, June-Koo~Kevin Rhee, and Francesco Petruccione.
\newblock Quantum classifier with tailored quantum kernel.
\newblock {\em npj Quantum Information}, 6(1):41, May 2020.

\bibitem{PARK2020126422}
Daniel~K. Park, Carsten Blank, and Francesco Petruccione.
\newblock The theory of the quantum kernel-based binary classifier.
\newblock {\em Physics Letters A}, 384(21):126422, 2020.

\bibitem{10.1145/3313276.3316310}
Ewin Tang.
\newblock A quantum-inspired classical algorithm for recommendation systems.
\newblock In {\em Proceedings of the 51st Annual ACM SIGACT Symposium on Theory of Computing}, STOC 2019, page 217--228, New York, NY, USA, 2019.
 Association for Computing Machinery.

\bibitem{tang2019quantuminspired}
Ewin Tang.
\newblock Quantum-inspired classical algorithms for principal component
  analysis and supervised clustering, 2019.
\newblock arXiv:1811.00414.

\bibitem{Arrazola2020quantuminspired}
Juan~Miguel Arrazola, Alain Delgado, Bhaskar~Roy Bardhan, and Seth Lloyd.
\newblock Quantum-inspired algorithms in practice.
\newblock {\em {Quantum}}, 4:307, August 2020.

\bibitem{holik_pattern_2018}
Federico Holik, Giuseppe Sergioli, Hector Freytes, and Angelo Plastino.
\newblock Pattern recognition in non-kolmogorovian structures.
\newblock {\em Foundations of Science}, 23(1):119--132, 2018.

\bibitem{Plos}
Giuseppe Sergioli, Roberto Giuntini, and Hector Freytes.
\newblock A new quantum approach to binary classification.
\newblock {\em PLOS ONE}, 14(5):1--14, 05 2019.

\bibitem{IJTP}
Giuseppe Sergioli, Gustavo~Martin Bosyk, Enrica Santucci, and Roberto Giuntini.
\newblock A quantum-inspired version of the classification problem.
\newblock {\em International Journal of Theoretical Physics},
  56(12):3880--3888, 2017.

\bibitem{IJQI}
Giuseppe Sergioli, Giorgio Russo, Enrica Santucci, Alessandro Stefano,
  Sebastiano~Emanuele Torrisi, Stefano Palmucci, Carlo Vancheri, and Roberto
  Giuntini.
\newblock Quantum-inspired minimum distance classification in a biomedical
  context.
\newblock {\em International Journal of Quantum Information}, 16(08):1840011,
  2018.

\bibitem{sergioli_quantum-inspired_2021}
Giuseppe Sergioli, Carmelo Militello, Leonardo Rundo, Luigi Minafra, Filippo
  Torrisi, Giorgio Russo, Keng~Loon Chow, and Roberto Giuntini.
\newblock A quantum-inspired classifier for clonogenic assay evaluations.
\newblock {\em Scientific Reports}, 11(1):2830, 2021.

\bibitem{hastie_elements_2009}
Trevor Hastie, Robert Tibshirani, and Jerome Friedman.
\newblock {\em The Elements of Statistical Learning: Data Mining, Inference,
  and Prediction, Second Edition}.
\newblock Springer Series in Statistics. Springer New York, 2009.

\bibitem{SS}
Enrica Santucci and Giuseppe Sergioli.
\newblock Classification problem in a quantum framework.
\newblock In Andrei Khrennikov and Bourama Toni, editors, {\em Quantum
  Foundations, Probability and Information}, pages 215--228. Springer
  International Publishing, Cham, 2018.

\bibitem{watrous_2018}
John Watrous.
\newblock {\em The Theory of Quantum Information}.
\newblock Cambridge University Press, 2018.

\bibitem{Helstrom1969}
Carl~W. Helstrom.
\newblock Quantum detection and estimation theory.
\newblock {\em Journal of Statistical Physics}, 1(2):231--252, Jun 1969.

\bibitem{Barnett:09}
Stephen~M. Barnett and Sarah Croke.
\newblock Quantum state discrimination.
\newblock {\em Adv. Opt. Photon.}, 1(2):238--278, Apr 2009.

\bibitem{Bae_2015}
Joonwoo Bae and Leong-Chuan Kwek.
\newblock Quantum state discrimination and its applications.
\newblock {\em Journal of Physics A: Mathematical and Theoretical},
  48(8):083001, jan 2015.

\bibitem{BENNETT20147}
Charles H. Bennett and Gilles Brassard.
\newblock Quantum cryptography: Public key distribution and coin tossing.
\newblock {\em Theoretical Computer Science}, 560:7--11, 2014.
\newblock Theoretical Aspects of Quantum Cryptography, celebrating 30 years
 of BB84.

\bibitem{Ho}
Roger~A. Horn and Charles~R. Johnson.
\newblock {\em Matrix Analysis}.
\newblock Cambridge University Press, 2 edition, 2012.

\bibitem{Fri}
A.~Friedman.
\newblock {\em Foundations of Modern Analysis}.
\newblock Dover Books on Mathematics Series. Dover, 1982.

\bibitem{Wooters_94}
Paul Hausladen and William K. Wootters.
\newblock A {\em pretty good} measurement for distinguishing quantum states.
\newblock {\em Journal of Modern Optics}, 41(12):2385--2390, 1994.

\bibitem{Barnum_Knill_2002}
H.~Barnum and E.~Knill.
\newblock Reversing quantum dynamics with near-optimal quantum and classical
  fidelity.
\newblock {\em Journal of Mathematical Physics}, 43(5):2097--2106, 2002.

\bibitem{Rus}
Mary~Beth Ruskai.
\newblock {\em Beyond Strong Subadditivity? Improved Bounds on the Contraction
  of Generalized Relative Entropy}, pages 350--366.
\newblock World Scientific Publishing, 1994.

\end{thebibliography}

\end{document}